\documentclass[nolinenumbers,cleveref]{lipics-v2021}
\usepackage{listings}
\usepackage{xspace}
\usepackage{eurosym}
\usepackage[loadshadowlibrary]{todonotes}
\newcommand{\var}[1]{\text{\lstinline+#1+}}

\newtheorem*{theorem*}{Theorem}

\def\cA{\ensuremath{{\mathcal A}}}   \def\cB{\ensuremath{{\mathcal B}}}   \def\cC{\ensuremath{{\mathcal C}}}

\newcommand{\set}[1]{\left\{ #1 \right\}}

\mathchardef\mhyphen="2D

\newcommand{\eg}{\emph{e.g.}\xspace}
\newcommand{\etal}{\emph{et al.}\xspace}

\usepackage{epsfig}

\newcommand{\eur}{\text{\euro{}}}
\newcommand{\usd}{{\$}}
\newcommand{\gbp}{{\text{\pounds}}}

\newcommand{\Reals}{\ensuremath{\mathbb{R}}}
\graphicspath{{./pictures/}}

\title{Defensive Rebalancing in Networks of Automated Market Makers}
\titlerunning{Defensive Rebalancing in Networks of AMMs} 
\authorrunning{S. Devorsetz and M. Herlihy}
\author{Sam Devorsetz}{Computer Science Dept., Brown University, Providence, USA}{alexander_devorsetz@alumni.brown.edu}{}{}
\author{Maurice Herlihy}{Computer Science Dept.,rown University, Providence, USA}{maurice.herlihy@gmail.com}{https://orcid.org/0000-0002-3059-8926}{}
\ArticleNo{33}
\Copyright{Anonymous Authors}
\ccsdesc[500]{Information systems~Distributed storage} 
\keywords{Automated Market Makers, Arbitrage, Convex Optimization} 
\category{Conference Submission} 

\begin{document}
\nolinenumbers

\maketitle
\thispagestyle{empty} 

\begin{abstract}
This paper introduces and analyzes \emph{defensive rebalancing},
a novel mechanism for preventing arbitrageurs from extracting surplus from networks of constant-function market makers (CFMMs).

A \emph{rebalancing} transfers assets directly from one CFMM's pool to another's,
bypassing the CFMMs' standard trading protocols.
In any \emph{arbitrage-prone} configuration,
we prove there exists a rebalancing to an \emph{arbitrage-free} configuration
that strictly increases the liquidity of at least one participating CFMM 
without decreasing the liquidity of any other participating CFMM.
Moreover,
we prove that a configuration is arbitrage-free if and only if it is
\emph{Pareto efficient} under rebalancing.

We prove that for any log-concave trading function, 
including the ubiquitous constant-product market maker, 
finding an optimal arbitrage-free rebalancing that maximizes the aggregate liquidity of the participating CFMMs while preserving each participant's liquidity 
can be formulated as a convex optimization problem 
with a unique, computationally tractable solution.

We extend this framework to \emph{mixed rebalancing},
where a subset of participating CFMMs
use a combination of direct transfers and standard trades
to transition to an arbitrage-free configuration 
while capturing arbitrage opportunities from non-participating CFMMs 
and external price-setting oracles such as centralized exchanges.

Our results provide a rigorous foundation for future CFMM protocols that coordinate defensive rebalancing to internalize arbitrage surplus within networks of CFMMs.
\end{abstract}
\newpage
\section{Introduction}
\label{intro}
A \emph{constant-function automated market maker}\cite{AngerisC2020} (CFMM)
$\cA$ is an automaton that trades one token for another.
It owns two pools of tokens,
let's call them dollar tokens (\usd) and euro tokens (\eur).
$\cA$ holds $x_1$ dollars in its dollar pool,
and $x_2$ euros in its euro pool.
Token exchange rates are determined by a \emph{trading function} $F(x_1, x_2)$,
satisfying some common-sense technical conditions given below.
Initially $F(x_1, x_2) = k$, for some $k > 0$, called $\cA$'s \emph{liquidity}.
A trader who transfers $\delta_\usd$ dollars to $\cA$ receives $\delta_\eur$ euros in return,
where $F(x_1+\delta_\usd,x_2-\delta_\eur) = k$.
That is, a \emph{trade} changes $\cA$'s pool sizes but leaves its
liquidity unchanged.

Consider the following \emph{circular trade} scenario.
There are three CFMMs,
each trading one token type for another.
$\cA$ trades dollars for euros,
$\cB$ trades euros for pounds (\gbp),
and $\cC$ trades pounds for dollars.
All three have trading function $F(x_1, x_2) := x_1 \cdot x_2 = 3$.

\begin{itemize}
  \item $\cA$ has \eur 1 and \usd 3,
  \item $\cB$ has \gbp 1 and \eur 3, and
  \item $\cC$ has \usd 1 and \gbp 3.
\end{itemize}

Alice (an \emph{arbitrageur}) sees an opportunity to make a quick risk-free profit.
She sells \usd 1 to $\cC$ for \gbp $\frac{3}{2}$,
which she sells to $\cB$ for \eur $\frac{9}{5}$,
which she sells to $\cA$ for \usd $\frac{27}{14}$,
collecting a profit of \usd $\frac{13}{14}$.
Her profit comes at the expense of the \emph{liquidity providers} (LPs),
the investors who loaned their tokens to $\cA$, $\cB$, and $\cC$.

Suppose instead that before Alice strikes,
the LPs for $\cA, \cB$ and $\cC$, recognizing their danger,
agree to transfer \eur 1 directly from $\cB$'s euro pool to $\cA$'s euro pool,
\gbp 1 from $\cC$'s pound pool to $\cB$'s pound pool,
and \usd 1 from $\cA$'s dollar pool to $\cC$'s dollar pool.
These are direct pool-to-pool transfers,
bypassing the CFMMs' usual trading protocols.

The state resulting from this \emph{rebalancing} has
three interesting properties.
First, arbitrage across $\cA,\cB,\cC$ has become impossible:
no circular trade results in a profit.
Second, rebalancing has improved the liquidity levels of
\emph{all three} CFMMs (each went from 3 to 4),
without reducing the value of any CFMM's assets (in a sense formalized below).
Third,
no subsequent rebalancing can increase one CFMM's liquidity without reducing another's.
The first condition improves social welfare by harmonizing exchange rates;
the second improves individual welfare because higher liquidity
implies lower price slippage, making that CFMM more attractive to traders.
The third condition means individual welfare cannot be further improved.

This paper shows that the first and third conditions are equivalent,
and the second follows as a consequence of the others.
As explained in \Cref{model},
a \emph{configuration} is the global state of a network of CFMMs.
A configuration is \emph{arbitrage-free} if no agent
can make a profit simply by executing a sequence of trades.
As noted, a \emph{rebalancing} is a set of direct transfers from one CFMM pool 
directly to another.
A configuration $c$ is \emph{Pareto efficient} (under rebalancing)
if any rebalancing starting from $c$ that increases one CFMM's liquidity must reduce another's.
In \Cref{pareto}, we show that
\emph{the arbitrage-free configurations are
precisely the Pareto-efficient configurations}.
In short, the Pareto-efficient configurations
are exactly those that eliminate arbitrage opportunities while
preserving or improving every participating CFMM's liquidity.

This equivalence has several implications.
\begin{itemize}
\item
  From any arbitrage-prone configuration,
  there exist rebalancings that simultaneously eliminate arbitrage exposure
  and strictly increase some CFMMs' liquidities without decreasing the others'.
\item
  One could provide ``rebalancing as a service'' where
  a set of participating CFMMs periodically check whether
  their configuration is arbitrage prone.
  If so, they agree on a rebalancing that simultaneously
  protects them from arbitrage
  and renders their liquidity levels Pareto efficient.

\item
  \Cref{convex} shows that for any log-concave trading function,
  the problem of finding a rebalancing leading to an arbitrage-free/Pareto-efficient
  configuration can be formulated as a convex optimization problem,
  computationally tractable using widely available solvers.
  Moreover, the equivalence between arbitrage-free and Pareto-efficient
  configurations simplifies the optimization program's formulation.
\item 
  \Cref{mixed} shows that this convex optimization program can be readily extended to
  systems where some CFMMs refuse to allow direct pool-to-pool transfers,
  and some act like powerful centralized exchanges,
  setting global market prices for certain tokens.
\end{itemize}

Our analysis assumes that a proposed rebalancing is executed
before the corresponding arbitrage opportunity has been exhausted by external traders. 
We focus on the existence and properties of such rebalancings rather than on the mechanisms required to schedule or prioritize them.

\section{Model}
\label{model}
A function $F: \mathbb{R}^k \to \mathbb{R}$ is \emph{strictly convex}
if for all $t \in (0,1)$ and distinct $x_1, x_2$,
\begin{equation*}
F(t x_1 + (1-t) x_2) < t F(x_1) + (1-t) F(x_2).
\end{equation*}
A positive-valued function $F$ is \emph{log-concave}
if for all $t \in (0,1)$ and distinct $x_1, x_2$,
\begin{equation*}
\log F(t x_1 + (1-t) x_2) > t \log F(x_1) + (1-t) \log F(x_2).
\end{equation*}
We use ``='' for equality and ``:='' for definitions.

\subsection{Automated Market Makers}
\emph{Agents} own \emph{tokens}.
A token is any valuable item that can be transferred electronically.
An \emph{automated market maker} (AMM) is an automaton that
trades between two token types, $T_1, T_2$.
An AMM manages a \emph{pool} of each token.
An AMM's \emph{state} is a pair $(x_1, x_2)$,
meaning the AMM has $x_1$ units of $T_1$ in its first pool,
and $x_2$ units of $T_2$ in its second pool.\footnote{Our results generalize to AMMs that trade more than two types of tokens, but we focus on two-pool AMMs for simplicity.}
A \emph{configuration} $c$ is an ordered collection of AMM states.

We consider \emph{constant-function market makers}~\cite{AngerisC2020} (CFMMs)
in which each trade preserves an invariant of the form $F(x_1, x_2) = k$,
where $k > 0$, and $F: \mathbb{R}^2 \to \mathbb{R}$
is twice differentiable, strictly increasing in each argument, and log-concave.
The value $k$ is called the CFMM's \emph{liquidity}.\footnote{This term is sometimes called the CFMM's \emph{invariant} or \emph{liquidity constant}. In our model, however, it is neither invariant nor constant over the course of repeated rebalancings.}
Perhaps the most popular CFMM is the
\emph{constant-product market maker} (CPMM) of the form
$F(x_1, x_2) := x_1 \cdot x_2 = k$.
Uniswap~\cite{uniswapv3}, Balancer~\cite{balancer}, Bullish~\cite{Bullish},
CoinEx~\cite{CoinEx2025}, and Skate~\cite{SkateAMM2025}
are examples of commercial CPMMs.

A \emph{liquidity provider} (LP) is an agent who invests in a CFMM
by loaning tokens to its liquidity pools.
LPs expect to profit from trading fees,
and typically receive governance rights over the CFMM.

For an AMM holding $x_\usd$ dollars and $x_\eur$ euros,
a \emph{trade} is a pair $(\delta_\usd, \delta_\eur)$ such that
\begin{equation*}
F(x_\usd + \delta_\usd, x_\eur - \delta_\eur) = F(x_\usd, x_\eur).
\end{equation*}
Thus, if a trader sends $\delta_\usd$ dollars to the AMM,
it receives $\delta_\eur$ euros in return,
where $\delta_\eur$ is chosen so that the liquidity remains unchanged.
Formally,
the AMM's possible states under trading
correspond to a \emph{level set} of $F$:
\begin{equation*}
\{ (x_1, x_2) \mid F(x_1, x_2) = k \},
\end{equation*}
where the liquidity $k$ remains constant for the duration of the trade.

For a CFMM with trading function $F(x_1, x_2)$, the \emph{spot price} of
token $T_1$ in terms of $T_2$ in state $(x_1, x_2)$ is:
\begin{equation*}
-\frac{\partial F / \partial x_1}{\partial F / \partial x_2}(x_1, x_2).
\end{equation*}
Trades change spot prices, an effect known as \emph{slippage}.
For example,
a CPMM $x_1 x_2 = 2$ in state $(1, 2)$ has a dollar spot price of $2$ euros,
but a trader who deposits $2$ euros changes the state to $(\frac{1}{2}, 4)$,
and receives only $\frac{1}{2}$ dollar.
The trade increased the spot price as dollars become more scarce compared to euros.

Consider two CFMMs $\cA, \cB$, both trading dollars against euros.
If $\cA$ and $\cB$ have different spot prices,
then they are vulnerable to \emph{arbitrage}.
Suppose euros are cheaper in $\cB$ than in $\cA$.
An \emph{arbitrageur} agent spends $d$ dollars
to buy euros from $\cB$,
then sells those euros to $\cA$ where they are more valuable,
receiving $d' > d$ dollars in return.
The agent has made a risk-free profit at the expense of the CFMMs' LPs.
Arbitrage can of course involve multiple AMMs.

\begin{definition}
    A configuration is \emph{arbitrage-prone} if an agent starting
    with a basket of tokens $(b_1, \ldots, b_m)$ can execute a sequence
    of trades and end up with a basket $(b_1', \ldots, b_m')$
    where $b_i' \geq b_i$, and at least one inequality is strict.
    A configuration that is not arbitrage-prone is \emph{arbitrage-free}.
\end{definition}

A \emph{rebalancing} transfers tokens directly between liquidity pools
belonging to different CFMMs.
Unlike trades,
rebalancing typically changes both AMMs' liquidities and their spot prices.
A rebalancing is self-funding in the sense that it does not change
the total amounts of each token held in participants' pools.
A rebalancing \emph{improves liquidity}
if every AMM's post-rebalancing liquidity is at least its original liquidity, and at least one such inequality is strict.

Note that trades and rebalancings are distinct operations.
A trade takes place between a party and a CFMM, and always leaves the AMM's liquidity unchanged.
A rebalancing takes place among multiple CFMMs, and typically alters their liquidities.
Trades are a standard part of classical AMMs,
while rebalancings are (as far as we know) analyzed for the first time in this paper.

Formally,
a rebalancing is defined as follows.
Starting from a configuration
\begin{equation*}
c := (x_{1,1}, x_{1,2}), \ldots, (x_{n,1}, x_{n,2}),
\end{equation*}
a rebalancing carries $c$ to a new configuration
\begin{equation*}
c' := (x_{1,1}', x_{1,2}'), \ldots, (x_{n,1}', x_{n,2}').
\end{equation*}
Tokens can be transferred only between pools holding the same token type.
We encode permitted transfers in the index set:
\begin{equation*}
\mathcal{T} := \left\{ (i,j,k,\ell) \mid i < k \text{ and the pools corresponding to } x_{i,j}, x_{k,\ell} \text{ have the same token type} \right\}.
\end{equation*}
A rebalancing is defined by a collection:
\begin{equation*}
\left\{ \delta_{i,j,k,\ell} \mid (i,j,k,\ell) \in \mathcal{T} \right\}.
\end{equation*}
Each $\delta_{i,j,k,\ell}$ represents a transfer of $|\delta_{i,j,k,\ell}|$ tokens.
If $\delta_{i,j,k,\ell} > 0$, then $|\delta_{i,j,k,\ell}|$ tokens are transferred
from the pool corresponding to $x_{i,j}$ to the pool corresponding to
$x_{k,\ell}$; if $\delta_{i,j,k,\ell} < 0$, the transfer is in the opposite
direction.
More precisely, the rebalancing sends each pool $x_{i,j}$ to $x_{i,j}'$, where:
\begin{equation*}
x_{i,j}' = x_{i,j} + \sum_{\substack{k,\ell:\\(k,\ell,i,j)\in \mathcal{T}}} \delta_{k,\ell,i,j} - \sum_{\substack{k,\ell:\\(i,j,k,\ell)\in \mathcal{T}}} \delta_{i,j,k,\ell}.
\end{equation*}
A rebalancing is \emph{feasible} if all final pool balances are positive.
All rebalancings considered here are assumed to be feasible.

\begin{definition}
A configuration $c$ is \emph{Pareto efficient} if no rebalancing starting from $c$ can strictly increase the liquidity of any CFMM without strictly reducing the liquidity of another CFMM.
\end{definition}

Protecting a network of CFMMs from arbitrage is a public good that benefits the network as a whole,
although those benefits may be distributed unevenly among the individual CFMMs.
Because a given arbitrage-prone configuration typically admits a continuum of
Pareto-efficient rebalancings,
there is considerable freedom in selecting a target configuration.
For example, in \Cref{convex},
we consider a rebalancing that selects the Pareto-efficient configuration
maximizing the product of the squares of the CFMMs' resulting liquidities.
Each participating CFMM may communicate additional constraints to the rebalancing agent
(e.g., ``I must retain at least 100 euros'').
The agent, in turn, is free to exclude any CFMM whose individual constraints
make rebalancing infeasible.
Ultimately, however, participation is voluntary:
every CFMM retains veto power and may decline to participate
in any proposed rebalancing for any reason.

We assume each sequence of trades or rebalancings executes \emph{atomically},
so that no other agent's actions are interleaved with it.

Our results characterize the allocations that are achievable under rebalancing.
They do not address how a rebalancing transaction obtains priority over competing arbitrage transactions on a blockchain.

As discussed in \Cref{remarks},
rebalancing is more powerful than trading alone
for protection against arbitrage.
Standard CFMM trades do not improve liquidity.
Moreover, no self-funding sequence of trades can, in general, harmonize prices:
the agent executing those trades generally cannot eliminate arbitrage
while simultaneously breaking even.

\section{Arbitrage and Pareto Efficiency}
\label{pareto}
\begin{theorem}
  \label{arbitrage-prone}
  Arbitrage-prone configurations are not Pareto efficient under rebalancing.
\end{theorem}

\begin{proof}
  Let 
\begin{equation*}
c := (x_{1,1},x_{1,2}), \ldots, (x_{n,1},x_{n,2})
\end{equation*}
be an arbitrage-prone configuration.
We construct a rebalancing that leads to a configuration
\begin{equation*}
c' := (x_{1,1}',x_{1,2}'), \ldots, (x_{n,1}',x_{n,2}')
\end{equation*}
  where $F_i(x_{i,1},x_{i,2}) \leq F_i(x_{i,1}',x_{i,2}')$ for $i \in \set{1,\ldots,n}$,
  and at least one inequality is strict,
  contradicting the hypothesis that $c$ is Pareto efficient.
  
  Because $c$ is arbitrage prone by hypothesis,
  an arbitrageur starting with a basket of tokens $(b_1, \ldots, b_m)$
  can execute a sequence of trades,
  ending with a basket $(b_1', \ldots, b_m')$,
  where each $b_j' \geq b_j$, for $j \in \set{1,\ldots, m}$,
  and at least one inequality is strict.
  
  To construct a rebalancing,
  borrow a basket $(b_1, \ldots, b_m)$,
  transfer tokens among the pools in a way that mimics
  that arbitrageur's sequence of trades,
  producing $(b_1', \ldots, b_m')$,
  where each $b_j' \geq b_j$,
  and at least one inequality is strict.
  Repay the loan, leaving the basket $(b_1'-b_1, \ldots, b_m'-b_m)$.
  Recall that each of these trades leaves each CFMM's liquidity unchanged.
  
 For each token type $T_k$ for which $b_k' - b_k > 0$,
  choose any CFMM that trades token type $T_k$,
  and transfer $b_k'-b_k$ tokens of type $T_k$ to that pool.
  This transfer strictly increases that CFMM's liquidity.
  Repeating this procedure for every token type with $b_k'>b_k$ 
  produces a rebalancing in which every CFMM's liquidity is at least its
  original value and at least one CFMM's liquidity is strictly greater,
  contradicting the hypothesis that $c$ is Pareto efficient.
\end{proof}
\begin{definition}
A \emph{valuation} for a network of CFMMs trading token types
$T_1,\ldots,T_m$ is a vector
$V=(V_1,\ldots,V_m)\in\mathbb{R}_{>0}^m$,
where $V_i$ is the value assigned to one unit of token $T_i$.
\end{definition}

\begin{definition}
  A valuation $V=(V_1,\ldots,V_m)$ is a \emph{global valuation}
  for a configuration if, for every CFMM that trades token types $T_p$ and $T_q$,
  \begin{equation*}
    \frac{V_p}{V_q}=P_{pq},
  \end{equation*}
  where $P_{pq}$ is the spot price of $T_p$ measured in units of $T_q$.
\end{definition}

The proof of the following lemma appears in the appendix.
\begin{lemma}
  \label{lemma:valuation}
  Any arbitrage-free configuration has a global valuation.
\end{lemma}

The remainder of this section uses the global valuation to assign a linear value to each CFMM. 
We show that increasing a CFMM's liquidity necessarily increases its value, 
while the total value of all CFMMs is preserved under rebalancing.

In an arbitrage-free configuration with global valuation $V$,
define the \emph{value} of the CFMM state $(x_{i,1},x_{i,2})$,
which trades token types $T_j$ and $T_k$, by
\begin{equation*}
  \Phi_V(x_{i,1},x_{i,2})
  := V_jx_{i,1} + V_kx_{i,2}.
\end{equation*}
The tangent line to the level set
\begin{equation*}
\set{(X_1,X_2) \;|\; F_i(X_1,X_2) = F_i(x_{i,1},x_{i,2})}
\end{equation*}
at point $ (x_{i,1} ,x_{i,2})$ is:
\begin{equation}
\label{Li}
  L_i := \set{(X_1,X_2) \;|\; \Phi_V(X_1,X_2) = \Phi_V(x_{i,1} ,x_{i,2})}.
\end{equation}
The tangent $L_i$ has slope $-\frac{V_j}{V_k}$, 
and divides $\Reals^2_{>0}$ into two half-planes.
Points in the half-plane below $L_i$ have lower value than $\Phi_V(x_{i,1} ,x_{i,2})$
and points above have higher values.
\begin{lemma}
  \label{val}
  $\Phi_V(x_{i,1}' ,x_{i,2}') < \Phi_V(x_{i,1} ,x_{i,2})$
  if and only if $(x_{i,1}' ,x_{i,2}')$ lies in the half-plane below $L_i$.
\end{lemma}

\begin{proof}
  From \Cref{Li}.
\end{proof}

Every point with higher liquidity lies above $L_i$.
\begin{lemma}
  \label{above}
  The region $\set{(X_1,X_2)  \mid F_i(X_1,X_2) \geq F_i(x_{i,1}, x_{i,2})}$ lies entirely
  in the half-plane above $L_i$.
\end{lemma}

\begin{proof}
Since the upper level set of a log-concave function is convex, 
the tangent line $L_i$ is a supporting line.
\end{proof}

\begin{corollary}
  \label{increase}
  Any rebalancing that increases the liquidity $F_i(x_{i,1}, x_{i,2})$
  must also increase the value $\Phi_V(x_{i,1}, x_{i,2})$.
\end{corollary}

\begin{corollary}
  \label{decrease}
  Any rebalancing that decreases the value $\Phi_V(x_{i,1}, x_{i,2})$
  must also decrease the liquidity $F_i(x_{i,1}, x_{i,2})$.
\end{corollary}

\begin{theorem}
\label{arbitrage-free}
  Any arbitrage-free configuration is Pareto efficient under rebalancing.
\end{theorem}

\begin{proof}
  We will show that from an arbitrage-free configuration,
  any rebalancing that increases one CFMM's liquidity
  must decrease another's.
  
  Consider a rebalancing that carries each
  $(x_{i,1}, x_{i,2})$ to $(x_{i,1}', x_{i,2}')$,
  where $F_i(x_{i,1}', x_{i,2}') > F_i(x_{i,1}, x_{i,2})$.
  By \Cref{increase}, $\cA_i$'s value increases: 
  $\Phi_V(x_{i,1}', x_{i,2}') > \Phi_V(x_{i,1}, x_{i,2})$.
  
  Because rebalancing neither creates nor destroys tokens,
  the sum of the CFMMs' values is invariant:
  \begin{equation*}
    \label{sum-invar}
    \sum_{i=1}^n \Phi_V(x_{i,1}, x_{i,2})
    =
    \sum_{i=1}^n \Phi_V(x_{i,1}', x_{i,2}').
  \end{equation*}
  Therefore there exists some $j$ such that
\begin{equation*}
    \Phi_V(x_{j,1}',x_{j,2}') < \Phi_V(x_{j,1},x_{j,2}).
\end{equation*}
  By \Cref{decrease}, $F_j(x_{j,1}', x_{j,2}') < F_j(x_{j,1}, x_{j,2})$,
  implying that the rebalancing strictly lowered at least one CFMM's liquidity.
\end{proof}
\Cref{arbitrage-prone,arbitrage-free} immediately imply the main result:
\begin{theorem}
\label{main}
A configuration is arbitrage free if and only if it is Pareto efficient.
\end{theorem}

\section{Rebalancing as Convex Optimization}
\label{convex}
\begin{figure}[tbh]
\centering
\begin{minipage}{0.95\textwidth}
\hrule
\vspace{0.5em}
\small
\begin{tabular}{@{}ll@{}}
\textbf{Parameters:} & $n \in \mathbb{N}$ (Number of CFMMs) \\
& $x_{i,j} \quad i \in \{1,\ldots,n\}, j \in \{1, 2\}$ (Initial pool sizes) \\
& $\mathcal{T} = \{ (i,j,k,\ell) \mid \text{transfer allowed between pools } x_{i,j}, x_{k,\ell} \}$ \\[0.5em]
\textbf{Decision variables:} & $x_{i,j}' \quad i \in \{1,\ldots,n\}, j\in \{1, 2\}$ (Final pool sizes) \\
& $\delta_{i,j,k,\ell} \quad \text{Amount (+/-) transferred along } (i,j,k,\ell) \in \mathcal{T}$ \\[0.5em]
\textbf{Objective:} & $\max \quad L(x_{1,1}',x_{1,2}', \ldots, x_{n,1}',x_{n,2}') := \sum_{i=1}^n \log F_i(x_{i,1}',x_{i,2}')$ \\[0.5em]
\textbf{Subject to:} & $x_{i,j}' \geq 0 \quad \text{for } i \in \{1,\ldots,n\}, j \in \{1,2\}$ \quad (non-negativity) \\[0.3em]
& $\log F_i(x_{i,1},x_{i,2}) \leq \log F_i(x_{i,1}',x_{i,2}'), \quad i=1,\ldots,n$ \quad (liquidity non-reduction) \\[0.3em]
& $x_{i,j}' = x_{i,j} + \sum_{\substack{k,\ell:\\(k,\ell,i,j)\in \mathcal{T}}}\delta_{k,\ell,i,j} - \sum_{\substack{k,\ell:\\(i,j,k,\ell)\in \mathcal{T}}}\delta_{i,j,k,\ell}$ \quad (conservation)
\end{tabular}
\vspace{0.5em}
\hrule
\vspace{0.5em}
\caption{Liquidity Rebalancing Optimization}
\label{program}
\end{minipage}
\end{figure}

In this section,
we show how to find a unique rebalancing
to protect a configuration from arbitrage.
We require the solution to be unique because mutually mistrustful LPs 
must independently verify the result.
Any ambiguity would lead to missed rebalancing opportunities.

Specifically, we show that
finding an arbitrage-free rebalancing can be cast as a convex optimization problem.
Moreover, we can exploit \Cref{main} to simplify the optimization program's
constraints by explicitly seeking a
Pareto-efficient final configuration instead of 
a (possibly more complicated) arbitrage-free final configuration.

\subsection{Parameters}
Rebalancings can transfer tokens only between pools of the same type. 
We capture which transfers are permissible through the index set:
\begin{equation*}
\mathcal{T} := \{ (i,j,k,\ell) \mid \text{transfers allowed between pools } x_{i,j}, x_{k,\ell} \}.
\end{equation*}
The other parameters are the starting pool sizes:
\begin{equation*}
(x_{1,1},x_{1,2}), \ldots, (x_{n,1},x_{n,2}).
\end{equation*}

\subsection{Decision Variables}
There are transfer amounts
\begin{equation*}
\delta_{i,j,k,\ell}, \qquad (i,j,k,\ell)\in\mathcal{T}.
\end{equation*}
If $\delta_{i,j,k,\ell}>0$, then
$|\delta_{i,j,k,\ell}|$ tokens are transferred from the pool
corresponding to $x_{i,j}$ to the pool corresponding to $x_{k,\ell}$.
If $\delta_{i,j,k,\ell}<0$, the transfer is in the opposite direction.
If $\delta_{i,j,k,\ell}=0$, no transfer occurs.

The remaining decision variables are the new pool sizes:
\begin{equation*}
(x_{1,1}',x_{1,2}'), \ldots, (x_{n,1}',x_{n,2}').
\end{equation*}

\subsection{Constraints}
The \emph{non-negativity} constraints require new pool sizes to be non-negative:
\begin{equation*}
x_{i,j}' \geq 0 \quad \text{for } i \in \{1,\ldots,n\}, j \in \{1,2\}.
\end{equation*}
This system of linear inequalities defines a convex set.

The \emph{conservation} constraints define the relationship
between the pool sizes and transfers.
For $i \in \{1,\ldots,n\}, j \in \{1,2\}$,
\begin{equation*}
x_{i,j}' = x_{i,j} + \sum_{\substack{k,\ell:\\(k,\ell,i,j)\in \mathcal{T}}}\delta_{k,\ell,i,j} - \sum_{\substack{k,\ell:\\(i,j,k,\ell)\in \mathcal{T}}}\delta_{i,j,k,\ell}.
\end{equation*}
This system of linear equalities defines a convex set.

The \emph{liquidity non-reduction} constraints require that
rebalancing not decrease any liquidity levels:
\begin{equation*}
F_i(x_{i,1},x_{i,2}) \leq F_i(x_{i,1}',x_{i,2}'), \quad i \in \{1,\ldots,n\}.
\end{equation*}
This constraint, however, is generally not convex,
but it can be made convex through a log transform:
\begin{equation*}
\log F_i(x_{i,1},x_{i,2}) \leq \log F_i(x_{i,1}',x_{i,2}'), \quad i \in \{1,\ldots,n\}.
\end{equation*}
Since the left-hand side is constant and the right-hand side is concave, 
the feasible region defined by this inequality is convex.

\subsection{Objective Function}
The objective function can express goals that span individual CFMMs,
thereby promoting social welfare.

\begin{theorem}
  Let $L: \mathbb{R}^n \to \mathbb{R}$ be strictly increasing in each argument.
  Any configuration that maximizes $L(F_1(x_{1,1}',x_{1,2}'), \ldots, F_n(x_{n,1}',x_{n,2}'))$
  over feasible rebalancings is arbitrage-free.
\end{theorem}

\begin{proof}
  Suppose $c$ maximizes
  $L(F_1(x_{1,1},x_{1,2}), \ldots, F_n(x_{n,1},x_{n,2}))$ but is not arbitrage-free.
  By \Cref{main}, $c$ is not Pareto efficient,
  so there is a rebalancing that takes each $(x_{i,1},x_{i,2})$ in $c$ to $(x_{i,1}',x_{i,2}')$,
  where each $F_i(x_{i,1},x_{i,2}) \leq F_i(x_{i,1}',x_{i,2}')$,
  and at least one inequality is strict.
  It follows that
  $L(F_1(x_{1,1},x_{1,2}), \ldots, F_n(x_{n,1},x_{n,2})) < L(F_1(x_{1,1}',x_{1,2}'), \ldots, F_n(x_{n,1}',x_{n,2}'))$,
  contradicting the hypothesis.
\end{proof}

\begin{corollary}
Any configuration that maximizes a weighted sum or weighted product of liquidities
\begin{equation*}
\max \sum_{i=1}^n w_i F_i(x_{i,1}',x_{i,2}') \quad \text{or} \quad \max \prod_{i=1}^n F_i(x_{i,1}',x_{i,2}')^{w_i}
\end{equation*}
for $w_i > 0$ is arbitrage-free.
\end{corollary}

For the objective function, we choose to maximize the product of the CFMMs' liquidities.
Maximizing this objective is not directly a convex optimization problem,
but it can be made convex through a log transform:
\begin{equation*}
\max \quad L(x_{1,1}',x_{1,2}', \ldots, x_{n,1}',x_{n,2}') := \sum_{i=1}^n \log F_i(x_{i,1}',x_{i,2}').
\end{equation*}
Because each $\log F_i(\cdot,\cdot)$ is concave on its domain by hypothesis,
$L(\cdot,\cdot)$ is concave as the sum of concave functions.
\Cref{program} shows the complete convex optimization problem.

\begin{theorem}
  The optimization program in \Cref{program} is a convex optimization problem.
\end{theorem}

\begin{proof}
The claim follows because we have shown that 
(1) the objective function to be maximized is concave, and
(2) the feasible region defined by the constraints is convex.
\end{proof}

\begin{theorem}
  The optimization program in \Cref{program} has a unique solution.
\end{theorem}

\begin{proof}
Because each $\log F_i$ is strictly concave, 
the objective function is strictly concave. 
A strictly concave function has at most one maximizer over a convex feasible set.
\end{proof}

As noted earlier,
if CFMMs have additional convex constraints,
they can be incorporated into the optimization program,
though care must be taken that these additional constraints
do not render the problem infeasible.

\section{Mixed Populations}
\label{mixed}
So far we have considered rebalancing among a population of CFMMs,
all of whom are willing to rebalance with one another to
mitigate a common arbitrage vulnerability.

Some LPs might be unwilling to allow a rebalancing agent
to transfer tokens directly in and out of their CFMMs' pools.
We distinguish between \emph{active} CFMMs who are willing to allow
direct pool-to-pool transfers, and \emph{passive} CFMMs, who are not.
We assume there is at least one active CFMM.
Passive CFMMs cannot be ignored,
because their spot prices might expose the others to arbitrage losses,
but passive CFMMs can interact with a rebalancing agent only through standard trades,
not direct pool-to-pool transfers.

A \emph{price oracle} is a market maker that models an external market 
that trades at a fixed reference price.
For example, consider a centralized exchange with (effectively) inexhaustible liquidity and no slippage
that is willing to sell arbitrary amounts of a token at a constant price.
This section shows how rebalancing can be extended to accommodate 
both passive CFMMs and price oracle market makers.

\subsection{Active vs. Passive CFMMs}
Recall the circular trade scenario where 
$\cA, \cB, \cC$ have the constant-product trading function
$F(x_{i,1}, x_{i,2}) := x_{i,1} \cdot x_{i,2} = 3$,
in a configuration where
$\cA$ has \eur 1 and \usd 3,
$\cB$ has \gbp 1 and \eur 3, and
$\cC$ has \usd 1 and \gbp 3.
Because this configuration is arbitrage-prone,
the LPs for $\cA$ decide to rebalance,
but this time without cooperation from the LPs for $\cB, \cC$.
The rebalancing agent does the following:
\begin{itemize}
\item
  Transfer \usd 0.5 out of $\cA$,
\item
  Trade \usd 0.5 to $\cC$ for \gbp 1.0,
\item
  Trade \gbp 1.0 to $\cB$ for \eur 1.5, and
\item
  Transfer \eur 1.5 into $\cA$.
\end{itemize}
Note that each trade is compatible with that CFMM's trading function.
The resulting configuration is arbitrage-free.
Because $\cB, \cC$ performed only trades,
their liquidity levels remain unchanged,
but the liquidity for $\cA$ has increased from $3$ to $6.25$.
By actively rebalancing, $\cA$ has effectively
executed the threatened arbitrage and allocated the profit to itself.
In general, that profit would be spread across the active CFMMs.

A rebalancing is \emph{restricted} if it does not change passive CFMMs' liquidities.
A restricted rebalancing can be implemented by executing both direct transfers 
between active CFMMs and trades with any CFMMs,
while ensuring the process is self-funding by making sure 
that every token ends up in some CFMM's pool.

As before,
consider CFMMs $\cA_1, \ldots, \cA_n$,
with trading functions $F_1(x,y), \ldots, F_n(x,y)$.
The \emph{active} CFMMs $\cA_1, \ldots, \cA_p$ support arbitrary transfers
in and out of their pools, for $p > 0$,
while the \emph{passive} CFMMs
$\cA_{p+1}, \ldots, \cA_n$ support only standard trades,
that is, transfers that leave their liquidities unchanged.

The network of CFMMs forms a \emph{CFMM graph},
where the CFMMs are the vertices,
and there is an edge linking each pair of CFMMs
$\cA, \cB$ that trade the same token type.
Without loss of generality, we assume this graph is \emph{connected},
meaning that between any pair of vertices $\cA, \cB$,
there is a sequence $\cA_1, \ldots, \cA_k$ called a \emph{path},
where consecutive CFMMs are connected by an edge,
$\cA = \cA_1$, and $\cB = \cA_k$.

\begin{lemma}
  \label{connected}
  Let $T$ be a token type traded by some CFMM $\cA_1$.
  An agent in possession of $\epsilon > 0$ tokens of type $T$
  can increase the liquidity of some active CFMM while leaving
  all other liquidity levels unchanged.
\end{lemma}

\begin{proof}
  By induction on $\ell$,
  the length of the shortest path $\cA_1, \ldots, \cA_\ell$
  linking $\cA_1$ to an active CFMM $\cA_\ell$ in the CFMM graph.

  For the base case, $\ell = 1$, so $\cA_1$ is active.
  The agent simply transfers the $\epsilon$ tokens into $\cA_1$'s pool for $T$,
  increasing $\cA_1$'s liquidity while leaving the others unchanged.

  For the induction step, assume the result for paths shorter than $\ell > 1$.
  The agent first trades at $\cA_1$ the $\epsilon$ tokens of type $T$
  for $\epsilon' > 0$ tokens of type $T'$,
  leaving $\cA_1$'s liquidity unchanged.
  The agent now holds $\epsilon' > 0$ tokens of type $T'$ traded by $\cA_2$,
  a passive CFMM linked to an active CFMM by a path shorter than $\ell$.
  The claim follows from the induction hypothesis.
\end{proof}

\begin{theorem}
  \label{arbitrage-prone-restricted}
  Arbitrage-prone configurations are not Pareto efficient under restricted rebalancing.
\end{theorem}

\begin{proof}
  Let $c$ be an arbitrage-prone configuration.

  As in the proof of \Cref{arbitrage-prone},
  an arbitrageur starting with a basket of tokens $(b_1, \ldots, b_m)$
  can execute a sequence of trades starting from $c$,
  ending with the basket $(b_1', \ldots, b_m')$,
  where each $b_j' \geq b_j$, for $j \in \{1, \ldots, m\}$,
  and at least one inequality is strict.

  For each token type $T_k$ for which $b_k' - b_k > 0$,
  \Cref{connected} implies that the agent can transfer those surplus tokens
  to some active CFMM trading $T_k$, 
  increasing that CFMM's liquidity while
  leaving all other liquidity levels unchanged.
  Repeating this procedure for every profitable token type yields a
  rebalancing in which every CFMM's liquidity is at least its original
  value and at least one CFMM's liquidity is strictly greater,
  contradicting the hypothesis that $c$ is Pareto efficient.
\end{proof}

\begin{theorem}
\label{arbitrage-free-restricted}
  Any arbitrage-free configuration is Pareto efficient under restricted rebalancing.
\end{theorem}

\begin{proof}
  The claim follows immediately from \Cref{arbitrage-free}.
\end{proof}

The \emph{restricted rebalancing convex optimization program}
requires only one change to the optimization program of \Cref{program}:
instead of maximizing the product of \emph{all} CFMM liquidities,
the restricted program maximizes the product of
the active CFMM liquidities \emph{only} (replacing $n$ with $p$):
\begin{equation}
\label{restricted-objective}
\max \quad L(x_{1,1}', x_{1,2}', \ldots, x_{p,1}', x_{p,2}') 
:= \sum_{i=1}^p \log F_i(x_{i,1}', x_{i,2}').
\end{equation}
The passive CFMMs' liquidities are bound only by the liquidity non-reduction constraint:
\begin{equation*}
F_i(x_{i,1}, x_{i,2}) \leq F_i(x_{i,1}', x_{i,2}'), \quad i \in \{p+1, \ldots, n\}.
\end{equation*}  
Let $c' = (x_{1,1}', x_{1,2}'), \ldots, (x_{n,1}', x_{n,2}')$
be the optimal configuration produced by solving the restricted rebalancing program.
Note that because the program's objective function is strictly concave,
and the feasible region is convex, the solution $c'$ is unique.

We claim that in $c'$, the passive CFMMs' non-reduction constraints hold with equality.

\begin{lemma}
\label{strict}
\begin{equation*}
F_i(x_{i,1}, x_{i,2}) = F_i(x_{i,1}', x_{i,2}'), \quad i \in \{p+1, \ldots, n\}.
\end{equation*}
\end{lemma}

\begin{proof}
Assume by way of contradiction that in $c'$,
$F_i(x_{i,1}', x_{i,2}') > F_i(x_{i,1}, x_{i,2})$ for some passive $\cA_i$.
Because $F_i(\cdot, \cdot)$ is strictly increasing in both arguments,
we can transfer $\epsilon > 0$ tokens out of pool $x_{i,1}'$
without violating the program's constraints:
$F_i(x_{i,1}' - \epsilon, x_{i,2}') \geq F_i(x_{i,1}, x_{i,2})$.
By \Cref{connected},
we can use these tokens to increase the liquidity of some active CFMM
while leaving all other liquidity levels unchanged.
This action increases the objective function's value,
contradicting the hypothesis that $c'$ is optimal.
\end{proof}

\begin{lemma}
\label{arb-free}
Configuration $c'$ is arbitrage-free.
\end{lemma}

\begin{proof}
By \Cref{main},
it suffices to show that $c'$ is Pareto efficient.
Assume by way of contradiction that there is a rebalancing
taking $c'$ to a configuration
$c^\dagger = (x_{1,1}^\dagger, x_{1,2}^\dagger), \ldots, (x_{n,1}^\dagger, x_{n,2}^\dagger)$
where for $i \in \{1, \ldots, n\}$,
$F_i(x_{i,1}^\dagger, x_{i,2}^\dagger) \geq F_i(x_{i,1}', x_{i,2}')$,
and at least one inequality is strict.

If $c^\dagger$ strictly increases an active CFMM's liquidity,
then $c^\dagger$ has a higher objective function value than $c'$,
while continuing to satisfy the restricted program's other constraints,
contradicting the hypothesis that $c'$ is Pareto efficient.

If, on the other hand, no active CFMM's liquidity increases,
then every strict increase must occur at a passive CFMM.
By \Cref{strict}, however, every passive CFMM already satisfies its
liquidity non-reduction constraint with equality in $c'$.
Any strict increase in a passive CFMM's liquidity would continue to
satisfy the restricted program's constraints,
and by \Cref{connected} the additional liquidity could be transferred to an
active CFMM without decreasing any other liquidity levels,
contradicting the optimality of $c'$.
\end{proof}

In summary, the optimization program of \Cref{program} extends naturally
to encompass passive CFMMs simply by omitting their liquidities from
the objective function.

\subsection{Price Oracles}
A \emph{price oracle} models a centralized exchange that sets the ``market price'' for a token by trading at that price.
Recall the circular trade scenario with CFMMs $\cA, \cB, \cC$.
As before, $\cA$ and $\cB$ have constant-product trading functions
$F_i(x_{i,1}, x_{i,2}) := x_{i,1} \cdot x_{i,2} = 3$.
$\cC$, however, 
is an oracle market maker with unlimited liquidity that trades pounds and dollars at par.
Initially,
$\cA$ has \eur 1 and \usd 3,
$\cB$ has \gbp 1 and \eur 3, and
$\cC$ has inexhaustible amounts of dollars and pounds.
This configuration is arbitrage-prone,
so the LPs for $\cA$ alone decide to rebalance.
The rebalancing agent:
\begin{itemize}
\item
  Transfers $\usd (\sqrt{3}-1) \approx \usd 0.732051$ from $\cA$,
\item
  Trades $\usd (\sqrt{3}-1)$ for $\gbp (\sqrt{3}-1)$ at $\cC$,
\item
  Trades $\gbp (\sqrt{3}-1)$ for $\eur (3-\sqrt{3}) \approx \eur 1.26795$ at $\cB$,
  and
\item
  Transfers $\eur (3-\sqrt{3})$ to $\cA$.
\end{itemize}
In the end,
$\cA$ has pools $(4-\sqrt{3}, 4-\sqrt{3})$
and $\cB$ has $(\sqrt{3}, \sqrt{3})$.
$\cA$'s liquidity has increased from $3$ to $19 - 8\sqrt{3} \approx 5.14359$,
while the others, being passive, remain unchanged.
This final configuration is arbitrage-free, and hence Pareto-efficient.

To create a convex optimization program that encompasses
one or more oracle AMMs,
the key insight is that a price-setting oracle AMM
is simply a passive CFMM with linear trading function
$F_i(x_{i,1}, x_{i,2}) := a x_{i,1} + b x_{i,2}$, where $a, b > 0$,
initialized with sufficient reserves to ensure no pool is exhausted by the rebalancing.
(Alternatively, one could permit oracle AMMs to have negative reserve
balances in the optimization program.)

\section{Remarks}
\label{remarks}
\subsection{Alternatives to CFMMs}
Although we have focused on CFMMs,
the equivalence between arbitrage-free and Pareto-efficient
configurations extends to any financial system prone
to arbitrage where parties have distinct private utility functions.

For example, suppose Alice owns a basket containing one gold token and one silver token.
She subjectively values gold at 1 dollar per token and silver at 2 dollars per token,
and thus values her basket at 3 dollars.
Bob also owns a basket containing one gold token and one silver token.
He subjectively values gold at 2 dollars per token and silver at 1 dollar per token,
and thus also values his basket at 3 dollars.
Both parties will accept any trade that increases the perceived value of their baskets
according to their own valuations.

Suppose an arbitrageur offers Alice $((1+\epsilon)/2)$ silver tokens 
in exchange for her gold token. 
She will accept, because she values the silver she receives at $(1+\epsilon)$ dollars,
while she values the gold she gives up at only 1 dollar.
Her new basket is worth $(3+\epsilon)$ dollars according to her valuation.
Similarly, the arbitrageur can offer Bob $((1+\epsilon)/2)$ gold tokens in exchange for his silver token.
Bob will also accept, and he too will value his new basket at $(3+\epsilon)$ dollars.

Collectively, the arbitrageur receives one gold token and one silver token while paying out $((1+\epsilon)/2)$ of each. 
The arbitrageur therefore earns a profit of $((1-\epsilon)/2)$ gold tokens and $((1-\epsilon)/2)$ silver tokens.

By comparison, a rebalancing might exchange Alice's gold token for Bob's silver token, 
leaving both parties in an arbitrage-free,
Pareto-efficient configuration in which each values their basket at 4 dollars.

\subsection{Alternatives to Rebalancing}
As an alternative to rebalancing,
LPs might instead execute ordinary trades among the CFMMs to mitigate their exposure to arbitrage.
Here we use a simple example to illustrate why LPs might prefer
defensive rebalancing (as introduced here)
to trading.
In general, however,
defensive trading cannot be \emph{self-funding}:
the trading agent must either inject tokens into the system of CFMMs
or withdraw them.

Consider the following example.
Recall the circular trade scenario where $\cA, \cB, \cC$
are in the arbitrage-prone configuration
$\cA: (\eur 1, \usd 3)$,
$\cB: (\gbp 1, \eur 3)$, and
$\cC: (\usd 1, \gbp 3)$.
\Cref{trade-only} shows an optimization program
that seeks an arbitrage-free configuration.
It uses slack variables $\sigma_\usd, \sigma_\eur, \sigma_\gbp$
to represent the amount of each token the agent adds
or removes from $\cA, \cB, \cC$.
The objective function is to minimize the sum of the squares
of the slack variables.
If a self-funding solution were possible,
these variables would all be zero.

This optimization problem is not convex,
but it is small enough to be readily solved by off-the-shelf 
solvers\footnote{For this example, we used Mathematica's \texttt{NMinimize}.}.
In this case,
minimizing the sum of the squares of the slack variables 
requires transferring $\sigma_\usd = \sigma_\eur = \sigma_\gbp = 4 - 2\sqrt{3} \approx 0.535898$ 
units of each currency from the three CFMMs to the agent,
yielding a solution that is both less effective than rebalancing
(because the CFMMs permanently surrender assets)
and more complicated
(because the agent must retain custody of the withdrawn assets).

\begin{figure}[tbh]
\centering
\begin{minipage}{0.95\textwidth}
\hrule
\vspace{0.5em}
\small
\begin{tabular}{@{}ll@{}}
\textbf{Parameters:} & $(x_{1,1}, x_{1,2}), \quad (x_{2,1}, x_{2,2}), \quad (x_{3,1}, x_{3,2})$ \\[0.5em]
\textbf{Decision variables:} & $(x_{1,1}', x_{1,2}'), \quad (x_{2,1}', x_{2,2}'), \quad (x_{3,1}', x_{3,2}')$ \\
& $\sigma_\eur, \quad \sigma_\gbp, \quad \sigma_\usd$ \\[0.5em]
\textbf{Objective:} & $\min \quad \sigma_\eur^2 + \sigma_\gbp^2 + \sigma_\usd^2$ \\[0.5em]
\textbf{Subject to:} & $x_{1,1}' x_{1,2}' = x_{1,1} x_{1,2} = 3 \quad \text{(Invariant for } \cA\text{)}$ \\[0.2em]
& $x_{2,1}' x_{2,2}' = x_{2,1} x_{2,2} = 3 \quad \text{(Invariant for } \cB\text{)}$ \\[0.2em]
& $x_{3,1}' x_{3,2}' = x_{3,1} x_{3,2} = 3 \quad \text{(Invariant for } \cC\text{)}$ \\[0.2em]
& $x_{1,1} + x_{2,2} = x_{1,1}' + x_{2,2}' + \sigma_\eur$ \\[0.2em]
& $x_{2,1} + x_{3,2} = x_{2,1}' + x_{3,2}' + \sigma_\gbp$ \\[0.2em]
& $x_{3,1} + x_{1,2} = x_{3,1}' + x_{1,2}' + \sigma_\usd$ \\[0.2em]
& $x_{1,1}' x_{2,1}' x_{3,1}' = x_{1,2}' x_{2,2}' x_{3,2}' \quad \text{(Arbitrage-free condition)}$ \\[0.2em]
& $x_{i,j}' > 0 \quad \text{for } i \in \{1,2,3\}, j \in \{1,2\}$
\end{tabular}
\vspace{0.5em}
\hrule
\vspace{0.5em}
\caption{Arbitrage protection with trades alone.}
\label{trade-only}
\end{minipage}
\end{figure}

\section{Related Work}
\label{related}
The notion of a CFMM is due to Angeris \etal~\cite{AngerisC2020},
who define the notion of a CFMM, 
analyze how CFMMs can serve as decentralized price oracles,
and formulate the \emph{optimal arbitrage problem},
namely, how to exploit an arbitrage-prone configuration most profitably.
Angeris \etal~\cite{AngerisCEB2022} consider how to use convex
optimization to detect arbitrage-prone configurations,
and how best to exploit them,
a problem they call \emph{optimal routing}.
Our use of optimization complements theirs:
rather than exploiting arbitrage opportunities,
we compute optimal rebalancings that eliminate them.
Angeris \etal~\cite{AngerisEC2020} investigate what happens
when a CFMM displaces an external market as a price setter.

Milionis \etal~\cite{MilionisMRZ2022} propose \emph{loss-versus-rebalancing}
(LVR) as an improvement over divergence loss as a way of measuring
LP costs in configurations with a central reference market.
LVR measures the difference between the value of a hypothetical
investment strategy that responds optimally to market price changes
and the value of an LP's actual investment in a CFMM,
which responds only through arbitrage.
Defensive rebalancing promises to reduce LVR by allowing LPs to
adjust their pool sizes dynamically in response to price changes at other markets.

Danos \etal~\cite{DanosKP2021} formulate various arbitrage problems
in a network of AMMs as optimization programs.
Angeris \etal~\cite{AngerisAECB2021} describe how a variety of problems
associated with CFMMs can be formulated as optimization problems.
Wang \etal~\cite{WangCWZDW21} perform an empirical investigation
of arbitrage transactions in Uniswap V2,
reporting on the frequency of arbitrage opportunities,
how often they are exploited, and their profitability.
Bartoletti \etal~\cite{BartolettiCL2021} describe arbitrage between
AMM pairs as a two-party game.
Dimitri~\cite{Dimitri2024} analyzes how arbitrage can
affect CFMM parameters such as pool sizes, prices, and liquidity.
Milionis \etal~\cite{MilionisMR23a}
study how block time affects arbitrage profits. 

As a way to avoid arbitrage altogether,
Krishnamachari \etal~\cite{KrishnamachariFG2021} describe \emph{dynamic curves},
which allow AMMs to adjust their trading curves to match oracle-reported prices.

In classical economics,
the \emph{First Welfare Theorem} states that, under suitable assumptions, 
every competitive equilibrium is Pareto efficient, 
while the \emph{Second Welfare Theorem} states that, 
under additional convexity assumptions, 
every Pareto-efficient allocation can be supported as a competitive equilibrium 
after an appropriate redistribution of initial endowments~\cite[Chapter~16]{mascolell1995microeconomic}.
\Cref{arbitrage-prone,arbitrage-free} bear a strong resemblance to the classical welfare theorems
and might at first seem to be special cases.
The analogy seems close: CFMMs play the role of consumers, 
tokens the role of commodities, 
liquidity the role of utility, 
rebalancing the role of feasible reallocation, 
and a common token valuation the role of a competitive price vector. 
Nevertheless, a direct reduction to the classical welfare theorems is not immediate. 
In particular, CFMMs do not explicitly optimize utility 
subject to exogenously specified budget constraints, 
and the feasible reallocations are constrained by the network of token pools. 
Thus, establishing \Cref{arbitrage-prone,arbitrage-free} as formal consequences 
of those theorems would require a careful correspondence between the two models.

\section{Conclusion}
\label{conclusion}
\subsection{Transaction Costs and Fees}
In practice, CFMMs charge \emph{trading fees}.
A CFMM $\cA_i$ with trading function $F_i(x_1,x_2)$
retains a fraction $\gamma_i \in (0,1]$ of each incoming trade.
If a trader transfers $\delta_\usd$ dollars to $\cA_i$ and receives $\delta_\eur$ euros in return,
the pool invariants update according to $F_i(x_{i,1} + \gamma_i \delta_\usd, x_{i,2} - \delta_\eur) = k$.
(The un-discounted fraction $(1-\gamma_i)\delta_\usd$ of dollars is typically retained by LPs as fee revenue.)
For example, in Uniswap, $\gamma_i$ depends on the pool's fee tier.

Adding fees to the model requires a few changes to \Cref{program}.
There is a new parameter: each $\cA_i$ has a fee parameter $\gamma_i$.
The conservation constraint is split into two parts:
\begin{align*}
  \text{for } i \le p: \quad 
  & x_{i,j}' = x_{i,j} + \sum_{\substack{k,\ell:\\(k,\ell,i,j)\in \mathcal{T}}}\delta_{k,\ell,i,j} - \sum_{\substack{k,\ell:\\(i,j,k,\ell)\in \mathcal{T}}}\delta_{i,j,k,\ell}, \quad \text{(active conservation)} \\[0.5em]
  \text{for } i > p: \quad 
  & x_{i,j}' = x_{i,j} + \sum_{\substack{k,\ell:\\(k,\ell,i,j)\in \mathcal{T}}}\gamma_i \delta_{k,\ell,i,j} - \sum_{\substack{k,\ell:\\(i,j,k,\ell)\in \mathcal{T}}}\delta_{i,j,k,\ell}, \quad \text{(passive conservation)}
\end{align*}
The conservation equations for active CFMMs are unchanged
(we assume pool-to-pool transfers between active CFMMs do not incur trading fees),
while the conservation constraints for passive CFMMs reflect their respective fees.

Agents also typically incur a small fixed cost with each trade.
Angeris \etal~\cite{AngerisCEB2022} observe that incorporating
such costs in the optimization program transforms a convex program
into a mixed-integer convex program.
Exact solutions for such programs may be computationally expensive,
but efficient approximate solvers exist
(see the cited reference for details).

The optimization program could be extended in a number of ways.
It might impose additional constraints,
such as enforcing upper or lower bounds on pool sizes.
The objective function could be extended to take into account costs
such as inventory risk or gas fees.

\subsection{Practical Issues}
Rebalancing requires active CFMMs to move tokens directly from
one pool to another without changing LPs' ownership shares.
Existing CFMMs (\eg,~\cite{uniswapv3,balancer,Bullish,CoinEx2025,SkateAMM2025})
provide mechanisms
for LPs to add or remove liquidity,
but none currently provide the exact functionality
needed to support rebalancing.

Rebalancing also raises security concerns.
When presented with a rebalancing proposal,
it is trivial for a CFMM to check that its own liquidity does not decrease.
However, verifying that the proposal maximizes the global objective function is expensive in gas costs.
For this reason,
the LPs might instead rely on a trusted third party to manage rebalancing,
in the same way DEXes currently engage services like
Chainlink~\cite{chainlink} to provide oracles,
or LayerZero~\cite{layerzero} or Axelar~\cite{axelar}
to provide cross-chain messaging.
In the context of the Ethereum Virtual Machine,
one could instead install a deterministic solver as a
``precompiled contract''~\cite{precompiled},
thereby providing a low-cost way for a 
CFMM to verify that a proposed rebalancing is fair and optimal.

In practice, successful defensive rebalancing requires executing a rebalancing transaction before competing arbitrage transactions eliminate the opportunity. 
Mechanisms for achieving this goal, including transaction sequencing, private order flow, auctions, or proposer-builder separation, 
are beyond this paper's scope.
One direction for future work is the design of decentralized mechanisms that detect arbitrage opportunities, 
compute Pareto-efficient rebalancings, 
and obtain sufficient execution priority to realize those rebalancings 
in adversarial transaction-ordering environments.


\bibliographystyle{plainurl}
\bibliography{references,zotero}

\section{Appendix}
\label{appendix}

\textbf{\cref{lemma:valuation}}\\
Every arbitrage-free configuration induces a global valuation.

\begin{proof}
Suppose the CFMM network trades token types $(T_1,\dots,T_m)$, and let $p_{jk}>0$ denote the spot price of $T_j$ in units of $T_k$ at any CFMM trading the pair $(T_j,T_k)$. 
We first show that arbitrage freedom implies that the product of spot prices around every cycle is $1$.

Consider a cycle
\begin{equation*}
T_{i_1},T_{i_2},\dots,T_{i_r},T_{i_1}.
\end{equation*}
If
\begin{equation*}
p_{i_1i_2}p_{i_2i_3}\cdots p_{i_ri_1}>1,
\end{equation*}
then, by continuity of the CFMM trading functions,
a sufficiently small quantity of $T_{i_1}$ can be traded around the cycle,
returning a strictly larger quantity of $T_{i_1}$ and thereby yielding an arbitrage profit. 
Hence
\begin{equation*}
p_{i_1i_2}p_{i_2i_3}\cdots p_{i_ri_1}\le 1.
\end{equation*}
Applying the same argument to the cycle traversed in the opposite direction gives the reciprocal inequality, so
\begin{equation*}
p_{i_1i_2}p_{i_2i_3}\cdots p_{i_ri_1}=1.
\end{equation*}

Assume first that the token graph is connected.
Choose $T_1$ as a num\'eraire and set $V_1:=1$.
For any other token $T_j$, choose a path
\begin{equation*}
T_1=T_{i_0},T_{i_1},\dots,T_{i_r}=T_j  
\end{equation*}
and define
\begin{equation*}
V_j = \prod_{\ell=0}^{r-1} p_{i_{\ell+1}i_\ell}.
\end{equation*}
This definition is independent of the chosen path.
Indeed, if two paths from $T_1$ to $T_j$ produced different values,
then traversing one path forward and the other backward would form
a cycle whose product of spot prices was not $1$, a contradiction.

Finally, suppose a CFMM trades $T_j,T_k$.
Take any path from $T_1$ to $T_k$ and append the edge from $T_k$ to $T_j$.
By the definition of $V_j$,
\begin{equation*}
V_j=p_{jk}V_k,  
\end{equation*}
and therefore
\begin{equation*}
p_{jk}=\frac{V_j}{V_k}.  
\end{equation*}
Thus $V=(V_1,\dots,V_m)$ is a consistent valuation.

If the token graph has multiple connected components,
apply the above construction independently to each component,
choosing an arbitrary token in each component as a num\'eraire
and assigning it value $1$.
\end{proof}

\end{document}